\documentclass[10pt,a4paper]{article}
\usepackage[T1]{fontenc}
\usepackage{lmodern}
\usepackage[margin=24mm]{geometry}
\usepackage{amsmath,amssymb,amsthm,mathtools}
\usepackage{graphicx,booktabs,microtype}
\usepackage{xcolor,xspace,stackengine}
\usepackage[colorlinks=true,linkcolor=blue,citecolor=blue,urlcolor=blue]{hyperref}
\hypersetup{pdftitle={Continuous-Time Critic as a Lyapunov Function},
 pdfauthor={Pavel Osinenko},pdfsubject={Projected adaptation and cart-pole stabilization}}
\newcommand{\Diff}{\ensuremath{\mathcal{D}}}							
\newcommand{\eps}{{\varepsilon}}										
\newcommand{\clip}{{\text{clip}}}										
\newcommand{\diam}{{\text{diam}}}										
\DeclareMathOperator*{\argmin}{\text{arg\,min}}							
\newcommand{\nrm}[1]{\left\lVert#1\right\rVert}							
\newcommand{\diag}[1]{{\text{diag}}\left(#1\right)}						
\newcommand{\cl}[1]{\text{cl}\left({#1}\right)}							
\newcommand{\low}{{\text{low}}}											
\newcommand{\up}{{\text{up}}}											
\DeclareMathAlphabet{\mathbbold}{U}{dsrom}{m}{n}						
\DeclareMathAlphabet{\mathbboldzero}{U}{bbold}{m}{n}					

\newcommand{\R}{{\mathbb{R}}}											
\newcommand{\states}{\mathbb S}											
\newcommand{\actions}{\mathbb A}										
\newcommand{\policy}{\pi}												
\newcommand{\K}{\ensuremath{\mathcal{K}}\xspace}						

\renewcommand{\states}{\mathbb X}
\renewcommand{\actions}{\mathbb U}
\newcommand{\Th}{\Theta}
\newcommand{\norm}[1]{\nrm{#1}}
\newtheorem{theorem}{Theorem}
\newtheorem{proposition}[theorem]{Proposition}
\newtheorem{lemma}[theorem]{Lemma}
\theoremstyle{remark}
\newtheorem{remark}[theorem]{Remark}
\newcommand{\CalfCost}{1208.8}
\newcommand{\MpcCost}{419.83}
\newcommand{\CalfSettling}{8.406}
\newcommand{\MpcSettling}{3.8375}
\newcommand{\CalfTravel}{5.9557}
\newcommand{\MpcTravel}{1.794}
\newcommand{\BoxVariation}{0.084829}
\newcommand{\NearVariation}{0.022127}
\newcommand{\CpuRatio}{1.6703}

\title{Continuous-Time Critic as a Lyapunov Function:\\
Projected Adaptation and Cart-Pole Stabilization}
\author{Pavel Osinenko\thanks{Correspondence: \href{mailto:p.osinenko@gmail.com}{p.osinenko@gmail.com}.}\\[4pt]
\small Central University, 7 Gasheka St., bld.~1, Moscow, 123056, Russia\\
\small Center for Engineering Systems and Sciences,\\
\small 30 Bolshoy Boulevard, bld.~1, Moscow, 121205, Russia\\
\small Sirius University of Science and Technology,\\
\small 1 Olimpiyskiy Avenue, Sirius, 354340, Russia}
\date{October 6, 2026}
\begin{document}
\maketitle
\begin{abstract}
We study adaptive stabilization of known control-affine systems by treating the current parameterized critic itself as a Lyapunov candidate.
Its total derivative contains a parameter-velocity term that can compensate unfavorable state evolution.
A projection onto the intersection of a bounded admissible velocity set and one decay halfspace enforces parameter invariance and Lyapunov decrease simultaneously.
We give an exact pointwise feasibility test, a conditional state-stability theorem, and a counterexample showing why a positive critic family containing a stabilizing member does not imply persistent feasibility.
A convex endpoint constraint provides a distinct sampled implementation.
A generic periodic polynomial-factor critic with a local Riccati anchor is evaluated on a force-limited cart-pole without mechanical-energy features, replay training, waypoints, or controller switching.
The frozen controller passes stabilization and numerical decay checks on 600 sampled initial states across two draw sets; 400 matched continuous trajectories also pass.
Independent reintegration audits and a conservative local feasibility certificate are distinguished from a global or machine-arithmetic guarantee.
Compared with nonlinear model-predictive control, the measured implementation uses less CPU but has higher cost, settling time, and cart travel.
The results support a conditional adaptive stabilization mechanism.
\end{abstract}

\section{Introduction}
\label{sec_introduction}
Control-Lyapunov functions (CLFs) convert the stabilization problem into feedback selection once an appropriate state function is available.
The classical universal construction of Sontag~\cite{Sontag1989UniversalConstr} makes this dependence explicit.
The difficult part is often finding a suitable CLF, particularly when a locally stabilizing quadratic function does not describe the desired large excursion.
Adaptive CLFs~\cite{Krstic1995ControlLyapunov} address nonlinear stabilization with unknown plant parameters through a different structured design problem.

The critic-as-Lyapunov-function (CALF) viewpoint uses a learned or adjusted critic as the stability object.
Earlier work develops reinforcement-learning schemes with Lyapunov constraints~\cite{Osinenko2020ReinforcementL}, an actor-critic framework with environment stability guarantees~\cite{Osinenko2023ActorCriticFram}, and a model-free CALF agent incorporating a nominal stabilizer~\cite{Osinenko2024CriticAsLyapuno}.
The present construction is explicitly \emph{model-based}: the known vector fields enter the critic derivative, and the sample-and-hold implementation predicts the next state.
The online parameters describe a critic and not an estimate of unknown masses or other plant coefficients.
No optimal value representation, Bellman equality, or ideal parameter vector is assumed here.
No nominal stabilizers are assumed known either.

The key observation is simple but operationally important.
For a smooth critic $V(x,\theta)$, the term $\nabla_\theta V^\top\dot\theta$ appears in its total derivative.
One can therefore adapt the critic to meet a decay inequality while the state evolves.
A projection that only preserves parameter bounds is insufficient: the same projection must also respect this particular derivative direction.
This yields a general-purpose formula with a state-dependent feasibility condition, rather than an unconditional stabilizer for every stabilizable system.

The contributions are as follows: (i) a bounded destination-projection formulation with an explicit feasibility budget and conditional state-stability proof; (ii) an exact-flow sampled counterpart whose endpoint optimization is convex for a positive factorized critic; and (iii) a reproducible cart-pole study separating numerical success, local analytic guarantees, and failure modes.
The chain rule, convex projection, and Lyapunov argument are established tools; the contribution claimed here is their particular continuous-time CALF construction and its implementation and assessment rather than a new general Lyapunov theorem.
It should be stressed that persistent feasibility remains a substantive condition for theoretical guarantees.

\subsection{Related work and positioning}
\label{sub_related_work_and_positioning}
Continuous-time online critic adaptation has a substantial history in adaptive dynamic programming.
Vamvoudakis and Lewis~\cite{Vamvoudakis2010OnlineActorCrit} develop simultaneous actor and critic updates to approximate the infinite-horizon Hamilton--Jacobi--Bellman equation, with excitation conditions for critic convergence.
Dierks and Jagannathan~\cite{Dierks2010OptimalControlO} use a single online approximator for the value function and associated control, establishing uniform ultimate boundedness in the presence of approximation errors.
These results establish precedents for continuous learning during control and for obtaining feedback from a critic.
Our stabilization condition is imposed directly on the current critic's total derivative, without requiring convergence to an optimal value representation.

Parameter-dependent Lyapunov functions and adaptation laws are also established in adaptive control~\cite{Krstic1995ControlLyapunov}.
Lopez and Slotine~\cite{Lopez2022UniversalAdapti} explicitly account for the parameter-derivative term and dynamically scale adaptation to handle estimation transients, starting from a family of CLFs for uncertain plant models.
More directly, Sinaei and Ebeigbe~\cite{Sinaei2026SafeAdaptiveCon} modify a parameter adaptation law through a minimally invasive real-time QP with robust adaptive CLF and barrier constraints for unknown dynamics.
Thus filtering an adaptation law through Lyapunov constraints already has published precedents.
In the present construction, the parameters are freely chosen critic coefficients for a known plant, and a frozen critic need only satisfy uniform positivity; its decrease is enforced using bounded parameter motion whenever the joint constraints are feasible.

Lyapunov-constrained reinforcement learning addresses several different stability or safety notions.
Chow et al.~\cite{Chow2018ALyapunovBasedA} construct Lyapunov constraints for policies in constrained Markov decision processes with expected cumulative-cost limits.
Han et al.~\cite{Han2020ActorCriticRein} learn a Lyapunov critic and actor using a data-based stability condition for stochastic systems.
Khaniki et al.~\cite{Khaniki2026SafeAndAdaptive} project finite policy-parameter updates using a linearized expected Lyapunov decrement, while updating the Lyapunov-network parameters separately by gradient descent.
Here the hard continuous-time constraint acts on the critic's own parameter velocity through the exact chain-rule term $\nabla_\theta V^\top\dot\theta$, evaluated with the actually applied input.

A complementary line jointly synthesizes feedback laws and Lyapunov certificates before deploying the resulting controller.
Neural Lyapunov Control~\cite{Chang2019NeuralLyapunovC} combines learning with counterexample search and formal verification, while Neural Lyapunov Redesign~\cite{Mehrjou2021NeuralLyapunovR} alternates certificate estimation and policy improvement to enlarge the region of attraction.
Recent state and output feedback synthesis~\cite{Yang2024LyapunovStableN} and continuous-time Zubov-based training~\cite{Li2025TwoStageLearnin} use scalable neural-network verification to certify regions for the synthesized controllers.
Our critic continues to evolve along the physical trajectory; the reported large-excursion cart-pole evidence is numerical, supplemented by a separate local certificate.
The present study therefore supplies a different implementation and guarantee scope from formally verified controller synthesis.

Finally, combining Lyapunov decrease and set invariance in a QP is standard CLF/CBF methodology~\cite{Ames2017ControlBarrierF}.
For the augmented state $(x,\theta)$, one can interpret $v=\dot\theta$ as a virtual control input and $\{0\}\times\Th$ as the target set.
This interpretation identifies the established optimization machinery underlying our construction.
The specific contribution is the bounded destination constraint on critic motion, its exact pointwise compensation budget and persistent-feasibility boundary, together with the convex sampled endpoint formulation and reproducible case study.

\section{Projected Critic Adaptation}
\label{sec_projected_critic_adaptation}
\subsection{Plant and critic}
\label{sub_plant_and_critic}
Consider the known plant
\begin{equation}
	\label{eqn_plant}
 \dot x=f(x)+g(x)u,\qquad x\in\states\subseteq\R^n,\quad u\in\actions,
\end{equation}
where $\states$ is an open neighborhood of the target $0$, $f,g$ are locally Lipschitz, and $\actions$ is an admissible input set.
Let $\Th\subset\R^N$ be compact and convex, and $V$ be $\mathcal C^1$ on a neighborhood of $\states\times\Th$.
Assume uniform positive-definite bounds on the domain of interest:
\begin{equation}
	\label{eqn_bounds}
 \kappa_\low(\norm x)\le V(x,\theta)\le\kappa_\up(\norm x),
 \qquad \theta\in\Th.
\end{equation}
Here $\kappa_\low$ and $\kappa_\up$ belong to class $\K$: they are continuous, strictly increasing functions that vanish at zero.
These bounds measure state size uniformly over the critic parameters.
To ensure that solutions remain in a region where the assumptions hold, the stability results below additionally require a compact sublevel set contained in $\states\times\Th$.
For periodic states the bounds are understood relative to distance to the target on the state space; a wrapped angle is not a global smooth coordinate.

Choose a continuous positive-definite decay target $w$ and an input law (in other words, policy) $\policy(x,\theta)\in\actions$.
A possible choice is
\begin{equation}
	\label{eqn_greedy}
 \policy(x,\theta)\coloneqq\argmin_{u\in\actions}
  \{\nabla_xV(x,\theta)^\top g(x)u+u^\top R(x,\theta)u\},\quad R\succ0.
\end{equation}
For a scalar interval $\actions=[-\bar u,\bar u]$, this is the saturated expression $\policy=\clip_{[-\bar u,\bar u]}(-L_gV/(2R))$, where $L_gV\coloneqq\nabla_xV^\top g$ and $\clip_{[a,b]}(z)\coloneqq\min\{b,\max\{a,z\}\}$.
Other admissible laws can be used; the derivative constraint must always use the \emph{actually applied} input.

Define, with column gradients,
\begin{equation}
	\label{eqn_pa}
 p\coloneqq\nabla_\theta V,\qquad
 a\coloneqq\nabla_x V^\top(f+g\policy)+w(x).
\end{equation}
For any absolutely continuous parameter trajectory with velocity $v$,
\begin{equation}
	\label{eqn_identity}
 \dot V=\nabla_xV^\top(f+g\policy)+p^\top v,
 \qquad \dot V\le-w(x)\ \Longleftrightarrow\ p^\top v+a\le0.
\end{equation}
In the linear critic case $V=\theta^\top\varphi(x)$, $p=\varphi(x)$.
For $\policy=-g^\top\Diff\varphi^\top\theta$, the state contribution is $\theta^\top\Diff\varphi f-\theta^\top\Diff\varphi gg^\top\Diff\varphi^\top\theta$, where $\Diff\varphi$ denotes the Jacobian of $\varphi$ with respect to $x$.
Notice that there is no unknown plant parameter in this calculation: $\theta$ is the chosen critic parameter, and its velocity contributes the directly adjustable term $p^\top v$.

\subsection{One projection for two requirements}
\label{sub_one_projection_for_two_requirements}
Let $v_0(x,\theta)$ be a preferred adaptation velocity and $\tau>0$ a design time constant.
Let $M$ be a symmetric positive-definite metric matrix, defining $\norm z_M^2\coloneqq z^\top Mz$.
Choose the closest admissible velocity to $v_0$ in this weighted norm:
\begin{equation}
	\label{eqn_projection}
 \boxed{\begin{alignedat}{2}
 v^*(x,\theta)\coloneqq{}&\argmin_v\quad&&\tfrac12(v-v_0)^\top M(v-v_0)\\
 &\text{subject to}\quad&&\theta+\tau v\in\Th,\quad p^\top v+a\le0,\\
 \dot\theta={}&v^*(x,\theta).
 \end{alignedat}}
\end{equation}
The first constraint requires the destination $\theta+\tau v$ to lie in $\Th$.
Here $\tau$ limits the adaptation speed; it need not equal a numerical integration step.
By convexity, the entire segment from $\theta$ to this destination remains in $\Th$.
The resulting velocities belong to the tangent cone $\mathcal T_\Th(\theta)\coloneqq\cl{\{c(y-\theta):c\ge0,\ y\in\Th\}}$, where $\cl{\cdot}$ denotes closure.
Unlike the whole tangent cone, the admissible velocity set is bounded: $\norm v\le\diam(\Th)/\tau$, where $\diam(\Th)\coloneqq\sup_{y,z\in\Th}\norm{y-z}$.
When feasible, the compact convex constraint set and strictly convex objective give a unique optimizer.

Equivalently, define the admissible destination set
\begin{equation}
 \mathcal K_\tau(x,\theta)\coloneqq\Th\cap
 \{y:p^\top y\le p^\top\theta-\tau a\}.
\end{equation}
Its weighted projection $y^*\coloneqq\argmin_{y\in\mathcal K_\tau}\norm{y-(\theta+\tau v_0)}_M^2$ is the closest admissible destination in the norm defined by $M$.
Substitution gives $v^*=(y^*-\theta)/\tau$.
Along a solution for which this selection is defined,
\begin{equation}
	\label{eqn_invariance}
 \theta(t)=e^{-t/\tau}\theta(0)+\frac1\tau
       \int_0^t e^{-(t-s)/\tau}y^*(s)\,ds\in\Th.
\end{equation}
The inclusion follows from closed convexity: the nonnegative weights sum to one.
Parameter invariance and critic decrease thus arise from separate constraints of the same optimization.

\begin{proposition}[Pointwise feasibility]
	\label{prp_feasibility}
The constraint set in~\eqref{eqn_projection} is nonempty precisely when
\begin{equation}
	\label{eqn_budget}
 a\le\frac{p^\top\theta-\min_{y\in\Th}p^\top y}{\tau}.
\end{equation}
\end{proposition}
\begin{proof}
Substitute $v=(y-\theta)/\tau$ into the decay constraint and minimize its left side over the compact set $\Th$.
The minimum is attained.
\end{proof}
The right-hand side of~\eqref{eqn_budget} is the largest reduction $-p^\top v$ that an admissible parameter velocity can contribute at the current $(x,\theta)$.
It is nonnegative because $y=\theta$ is admissible, but it can be zero at a parameter boundary or when $p=0$.
At a fixed $(x,\theta)$, reducing $\tau$ allows larger parameter velocities and cannot reduce this available compensation.
Along a closed-loop trajectory, however, the faster adaptation also changes future states and parameters; it therefore need not prevent later infeasibility.

For a box $\Th=[\theta_-,\theta_+]$, put $l=(\theta_--\theta)/\tau$ and $r=(\theta_+-\theta)/\tau$.
The exact test is
\begin{equation}
	\label{eqn_box}
 a+\sum_{p_i>0}p_i l_i+\sum_{p_i<0}p_i r_i\le0.
\end{equation}
For diagonal $M=\diag{m_i}$, with $m_i>0$, introduce the nonnegative Lagrange multiplier $\lambda$ of the decay constraint $p^\top v+a\le0$.
Minimizing the Lagrangian over the component bounds gives
\begin{equation}
 v_i(\lambda)\coloneqq\clip_{[l_i,r_i]}(v_{0i}-\lambda p_i/m_i),\qquad\lambda\ge0.
\end{equation}
Define the scalar residual $F(\lambda)\coloneqq a+\sum_i p_i v_i(\lambda)$.
If $F(0)\le0$, the box projection $v(0)$ is already the optimizer.
Otherwise the decay constraint is active, and its multiplier satisfies $F(\lambda)=0$.
Each term $p_i v_i(\lambda)$ is continuous and nonincreasing, so the same holds for $F$.
After checking~\eqref{eqn_box}, increase an upper trial value of $\lambda$ until $F\le0$, then use bisection between zero and that value.
This summarizes the procedure.

\subsection{Why feasibility is not automatic}
\label{sub_why_feasibility_is_not_automatic}
Consider $\dot x=x+u$, $V(x,\theta)=\theta x^2$, $\Th=[0.1,1]$, and $u=-\partial_xV=-2\theta x$.
The interior parameter $\theta=0.75$ stabilizes the plant with $\dot V=-0.75x^2$.
Nevertheless, at $\theta=0.1$, the derivative is $0.16x^2+x^2v$.
Every invariant-box velocity satisfies $v\ge0$ there, making strict decay impossible for $x\ne0$.
Thus even a uniformly positive family containing a stabilizing member need not remain feasible from all admissible initial parameters.
Generally speaking, feasibility can be achieved for suitably chosen critic models.

Sequential projection also fails in general.
At a corner of $\R_+^2$, the decay halfspace $v_1+v_2\le1$ contains $(-1,2)$, but its projection onto the positive orthant is $(0,2)$, outside the halfspace.
This explains why an invariance projection cannot simply be appended to an otherwise successful adaptation formula.

\section{Conditional Stability and Sampled Implementation}
\label{sec_conditional_stability_and_sampled_implementation}
\begin{theorem}[State stability under feasible adaptation]
	\label{thm_ct}
Suppose~\eqref{eqn_bounds} holds, $w$ is continuous positive definite, and $f+g\policy$ is continuous, bounded on compact subsets, and zero at $x=0$.
Let $\Omega_b := \{(x,\theta)\in\states\times\Th:V(x,\theta)\le b\}$ be compact with state projection inside $\states$.
Suppose the adaptation law has absolutely continuous local solutions which admit continuation while remaining in compact subsets, and its constraints remain feasible on $\Omega_b$.
Then every such solution starting in $\Omega_b$ is forward complete, stays there, and has $x(t)\to0$.
The state target is Lyapunov stable uniformly over initial parameters for sufficiently small initial states whose sublevels are covered by these assumptions.
\end{theorem}

\begin{proof}
Equation~\eqref{eqn_invariance} preserves the parameter set.
The chain rule and the joint decay constraint give
\begin{equation}
	\label{eqn_integral}
 V(x(t),\theta(t))+\int_0^t w(x(s))\,ds\le V(x(0),\theta(0)).
\end{equation}
The trajectory therefore stays in the compact sublevel and extends for all positive time by the continuation hypothesis.
The uniform bounds give stability: choose $\delta$ so that $\kappa_\up(\delta)<\kappa_\low(\eps)$ and the associated initial sublevel is covered.
Then $\norm{x(0)}<\delta$ implies $\norm{x(t)}<\eps$, independently of $\theta(0)$.
On the compact invariant set, the plant velocity is bounded, so $x(t)$ is Lipschitz in time.
Uniform continuity of $w$ there makes $w(x(t))$ uniformly continuous.
It is nonnegative and integrable by~\eqref{eqn_integral}, hence tends to zero.
Positive definiteness on the compact state projection then implies $x(t)\to0$.
\end{proof}

\begin{remark}[No parameter convergence]
	\label{rem_no_param_conv}
	The theorem does not assert parameter convergence to any particular vector.
\end{remark}

\begin{remark}[Regularity is a separate hypothesis]
Uniqueness of each convex optimizer does not imply continuity of the state-dependent optimizer map.
Projected dynamical systems require explicit solution and regularity conditions~\cite{Hauswirth2021ProjectedDynami}.
For example, with $\dot x=-x$, $V=ax^2$, decay target $x^2$, and nominal $\dot a=2(2-a)$, the constraint at $a=1$ is $x^2\dot a\le x^2$.
With bounds permitting both velocities, projection gives $\dot a=1$ for $x\ne0$ but $\dot a=2$ at $x=0$.
Smooth data and a unique optimizer do not remove this discontinuity.
This example diagnoses regularity, not instability.
\end{remark}

\subsection{A distinct endpoint-constrained realization}
\label{sub_a_distinct_endpoint_constrained_realization}
Choose a hold length $h\le\tau$ (in a sample-and-hold control regime) and let $y=\Phi_h(x,\policy(x,\theta))$ be the exact held-input plant trajectory endpoint starting from $x$ under constant control $\policy(x,\theta)$.
Set $\theta^+ :=\theta+hv$.
In addition to the box and instantaneous constraints, impose
\begin{equation}
	\label{eqn_endpoint}
 V(y,\theta+hv)-V(x,\theta)
 +\frac h2\{w(x)+w(y)\}+h\delta q(x)\le0,
 \qquad\delta\ge0,\quad q\ge0.
\end{equation}
If $V(y,\cdot)$ is convex, this is a convex constraint in $v$ for the \emph{fixed} selected input.
The sampled filter is not merely Euler integration of~\eqref{eqn_projection}; it solves an additional endpoint requirement.
The predicted flow is independent of the parameter destination because the physical input is held during the interval.

\begin{lemma}[Exact sampled decay]
	\label{lem_sampled}
For fixed $h>0$, suppose all exact held endpoints exist, all updates satisfy the admissible velocity set and~\eqref{eqn_endpoint}, and the critic has uniform bounds and a compact sublevel containing the initial pair.
Then parameters remain in $\Th$, sampling states remain in that sublevel, and $x_k\to0$.
The state is Lyapunov stable at sampling instants.
\end{lemma}
\begin{proof}
Since $\theta+\tau v\in\Th$ and $h\le\tau$, the destination $\theta^+=(1-h/\tau)\theta+(h/\tau)(\theta+\tau v)$ belongs to $\Th$.
Telescoping~\eqref{eqn_endpoint} gives $\frac h2\sum_{k=0}^{N-1}\{w(x_k)+w(x_{k+1})\}\le V(x_0,\theta_0)$.
For fixed $h>0$, the nonnegative series $\sum_k w(x_k)$ is finite, hence $w(x_k)\to0$.
If $x_k$ did not converge to zero, some subsequence would stay at least a fixed distance $\eps>0$ from zero.
Compactness of the state projection would give a further subsequence converging to $\bar x\ne0$.
Continuity would then imply $w(\bar x)=0$, contradicting positive definiteness.
For stability, fix $\eps>0$ and choose $\rho>0$ such that $\kappa_\up(\rho)<\kappa_\low(\eps)$ and the corresponding initial sublevel is covered by the assumptions.
Then $\norm{x_0}<\rho$ implies, for every $k$,
\begin{equation*}
 \kappa_\low(\norm{x_k})\le V(x_k,\theta_k)
 \le V(x_0,\theta_0)\le\kappa_\up(\norm{x_0})
 <\kappa_\low(\eps).
\end{equation*}
Strict monotonicity of $\kappa_\low$ gives $\norm{x_k}<\eps$, uniformly over the admissible initial parameters.
\end{proof}
Full intersample convergence additionally follows if the held-flow excursion over $[0,h]$ tends uniformly to zero as the initial state tends to zero, uniformly over admissible parameters.
The local cart-pole construction below has this property.
Numerically predicted endpoints do not satisfy the lemma automatically: rigorous prediction-error enclosures or exact flow would be needed for a machine-level claim.

\section{A Generic Cart-Pole Critic}
\label{sec_a_generic_cart_pole_critic}
\subsection{Physical model and periodic lift}
\label{sub_physical_model_and_periodic_lift}
Let $s=(x,v,\vartheta,\omega)\in\R^2\times S^1\times\R$, with upright $\vartheta=0$ and force $|u|\le12$ N.
A point-mass pole on a massless rod has dynamics
\begin{align}
	\label{eqn_cartpole}
 \dot x&=v,& \dot\vartheta&=\omega,\nonumber\\
 \dot v&=\frac{u+m\ell\sin\vartheta\,\omega^2-mg\sin\vartheta\cos\vartheta}
                  {M+m\sin^2\vartheta},&
 \dot\omega&=\frac{g\sin\vartheta-\cos\vartheta\,\dot v}{\ell}.
\end{align}
The nominal parameters are $M=1$ kg, $m=0.2$ kg, $\ell=0.8$ m, $g=9.81$ m/s$^2$.
There is no friction, rail-position constraint, actuator lag, slew-rate limit, measurement noise, or parameter mismatch in this study.
The uncontrolled upright equilibrium is unstable.

Use the generic lift and its degree-two monomials
\begin{equation}
 z=(x,v/2,2\sin\vartheta,2(1-\cos\vartheta),\omega)^\top,
 \qquad \psi=(z_i,\tfrac14z_i z_j:1\le i\le j\le5)^\top\in\R^{20}.
\end{equation}
Coordinates are numerically scaled in the stated SI conventions; their scales are design choices.
Define
\begin{align}
 H(s)&=1-\exp[-\tfrac12\{x^2+v^2+8(1-\cos\vartheta)+\omega^2/4\}],\nonumber\\
 \xi_\kappa(s)&=(x,v,\sin\vartheta+\kappa(1-\cos\vartheta),\omega)^\top,
 \qquad \gamma(s)=e^{-(1-\cos\vartheta)/b^2}.
\end{align}
Obtain $P\succ0$ from the upright continuous-time algebraic Riccati equation with $Q_0=\operatorname{diag}(2,1,70,8)$ and $R_0=0.35$.
The local checks use the stored matrix without assuming that it solves the equation exactly in machine arithmetic.
Define
\begin{equation}
	\label{eqn_critic}
 V(s,L,\eta)=H(s)\{\varepsilon\norm z^2+\norm{L\psi}^2\}
                  +We^\eta\gamma(s)\xi_\kappa^\top P\xi_\kappa.
\end{equation}
Here $L\in[-3,3]^{5\times20}$, $\eta\in[-7,0.7]$, $\varepsilon=0.002$, $W=8$, $b=1.5$, and $\kappa\in\{-0.25,0.25\}$ is fixed after initialization.
There are 101 adaptive parameters.

\begin{proposition}
	\label{prp_positive}
The critic family~\eqref{eqn_critic} is uniformly positive definite and proper on the cylindrical state space.
In an upright coordinate chart, $V=We^\eta s^\top Ps+O(\norm s^3)$ uniformly over its parameter box.
\end{proposition}
\begin{proof}
The lift $z$ vanishes only at the target, and $H>0$ elsewhere.
Therefore $V\ge\varepsilon H\norm z^2>0$ away from the target, independently of the parameters.
As $\norm{(x,v,\omega)}\to\infty$, $H\to1$ and $\norm z\to\infty$; the angle is compact.
Near the target, $H=O(\norm s^2)$, $\psi=O(\norm s)$, $\xi_\kappa=s+O(\norm s^2)$, and $\gamma=1+O(\norm s^2)$.
The bounded parameter set makes all remainders uniform; $We^\eta$ has a strictly positive lower bound.
\end{proof}
It should be noted that this family is not structure-free: it uses the plant model, periodic coordinates, scales, and a local Riccati matrix.
However, its features contain no mechanical-energy error or predesigned swing-up trajectory.
The positive factorization is not an SOS certificate of a globally negative Lie derivative.
In particular, $\nabla_LV=2H(L\psi)\psi^\top$ can vanish when $L\psi=0$ even if $L\ne0$ and no coefficient lies on a box face.

\subsection{Initialization and feedback}
\label{sub_initialization_and_feedback}
Write $\xi_\kappa=T_\kappa z$ and let rows of $N_\kappa$ form an orthonormal basis of $\ker T_\kappa$.
Define
\begin{equation}
 C_\kappa=T_\kappa^\top PT_\kappa+
       \frac{\operatorname{tr}(T_\kappa^\top PT_\kappa)}4N_\kappa^\top N_\kappa.
\end{equation}
For each $\kappa\in\{-0.25,0.25\}$, compute an upper Cholesky factor $B_\kappa$ satisfying $B_\kappa^\top B_\kappa=C_\kappa$.
Set $L_0^{(\kappa)}\coloneqq[\sqrt5 B_\kappa/\norm{B_\kappa}_{\rm F}\ \ 0_{5\times15}]$ and $\eta_0=0$, where $\norm{\cdot}_{\rm F}$ is the Frobenius norm.
At the initial state $s_0$, evaluate $V(s_0,L_0^{(\kappa)},0)$ separately for these two candidates and retain the candidate with the smaller value.
If the values are equal, choose $\kappa=+0.25$.
The selected $\kappa$ remains fixed throughout the run; only $L$ and $\eta$ subsequently adapt.
Thus each sign is evaluated with its own initialized factor, and the selection occurs once before applying the controller.

The preferred adaptation is recovery toward the initialized parameters: $v_0=0.5(\theta_0-\theta)$, with $\tau=0.01$ s and $M=I$ in \eqref{eqn_projection}.
The required state decay is $w=0.01q$, where
\begin{equation}
	\label{eqn_q}
 q(s)=x^2+0.2v^2+12\vartheta_{\rm wrap}^2+\omega^2.
\end{equation}
The squared principal angle is continuous on $S^1$, although it is not differentiable at the antipode.
The critic remains smooth and periodic; the proofs require continuity, not differentiability, of $q$.

To moderate force variation, use the penalty
\begin{equation}
	\label{eqn_penalty}
 R=R_f+(1-H)(We^\eta R_0-R_f)+t_d HG_{\rm GN},
 \qquad R_f=0.01,\quad t_d=0.002\ \mathrm{s},
\end{equation}
in~\eqref{eqn_greedy}.
For a factor $C^\top C=P$, define residual vectors $r_G=\sqrt H(\sqrt\varepsilon z,L\psi)^\top$ and $r_A=\sqrt{We^\eta\gamma}\,C\xi_\kappa$; then $G_{\rm GN}\coloneqq\norm{\Diff_g r_G}^2+\norm{\Diff_g r_A}^2\ge0$.
Here $\Diff_g r\coloneqq(\Diff_s r)g(s)$ is the directional derivative along the input vector field, with the critic parameters held fixed.
The derivative $\Diff_g r_G$ has a continuous zero extension at the target; $\Diff_g r_A$ generally does not vanish there.
Since the input has no position or angle component, $\norm{\Diff_g r_A}^2=We^\eta\gamma g^\top Pg$ exactly.
The added penalty is $O(\norm s^2)$ and does not change the local gain:
\begin{equation}
	\label{eqn_localgain}
 \policy(s,\theta)=-Ks+O(\norm s^2),\qquad K=R_0^{-1}B^\top P,
\end{equation}
uniformly in boxed parameters.
This nonnegative first-derivative surrogate is not a Hessian bound and does not itself guarantee feasibility, smooth commands, or a force slew limit.

\subsection{Local persistent feasibility}
\label{sub_local_persistent_feasibility}
Let $A,B$ be the upright linearization, $A_c=A-BK$, $Q_c=-(PA_c+A_c^\top P)$, and $S_-=We^{-7}$.
Uniform Taylor bounds in~\eqref{eqn_localgain} make $s^\top Ps$ a common local Lyapunov function for every fixed or varying boxed parameter selection.
For a zero parameter velocity, the normalized critic derivative has leading term $-s^\top Q_c s$.
Thus positive definiteness of $Q_c-(\sigma+\zeta)Q_s/S_-$, with $\sigma=0.01$, $Q_s=\operatorname{diag}(1,0.2,12,1)$, and a suitable reserve $\zeta\ge0$, gives a uniform local feasible witness $v=0$.

For fixed hold $h$, let $A_h=e^{Ah}$, $B_h=\int_0^h e^{At}B\,dt$, and $F_h=A_h-B_hK$.
The analogous leading matrix for~\eqref{eqn_endpoint} is
\begin{equation}
	\label{eqn_samplematrix}
 D_h=P-F_h^\top\left(P+\frac{h\sigma}{2S_-}Q_s\right)F_h
       -\frac{h(\sigma/2+\delta+\zeta)}{S_-}Q_s.
\end{equation}
When this and $P-F_h^\top PF_h$ are positive definite, uniform remainders give a common sufficiently small invariant state ellipsoid on which the zero update is feasible.
Box invariance is enforced directly, so this argument does not require an exact projection minimizer: any verified feasible boxed update is admissible.
Standard local held-flow bounds cover the intersample state.
This proves local persistent feasibility for the ideal nominal law, not entry from a hanging initial condition.

For the specified critic, $h=0.002$ s and $\delta=10^{-4}$, conservative exact-rational remainder and matrix checks accept
\begin{equation}
	\label{eqn_radius}
 r_*=2^{-25},\qquad s^\top Ps\le3.56591397745\times10^{-16}.
\end{equation}
The sampled and continuous error-bound comparisons are, respectively, $0.0003194091<0.0003770750$ and $0.05970460<0.19093633$ (rounded displays).
The stored binary matrix coefficients are interpreted as exact rationals; matrix positivity is checked by exact elimination and matrix exponentials are enclosed with Taylor remainders.
The ancillary certificate and executable checker record and reproduce the unrounded constants and matrix; Appendix~\ref{sec_local} states the sufficient inequalities.
The radius is extremely conservative and is \emph{not} a practical estimate of the swing-up attraction region.
The numerical simulations and a floating-point solver are not certified by this exact-arithmetic calculation.

\section{Numerical Study}
\label{sec_numerical_study}
\subsection{Protocol and implementation}
\label{sub_protocol_and_implementation}
The selected controller was frozen before the primary 300-state draw.
Two hundred starts have independent uniform coordinates in $[-1,1]\times[-2,2]\times[-\pi,\pi]\times[-4,4]$, using seed 20261024.
One hundred near-hanging starts use independent uniforms $x\in[-0.3,0.3]$, $v\in[-0.5,0.5]$, $\vartheta\in\pi+[-0.3,0.3]$, $\omega\in[-0.7,0.7]$, seed 20261025.
The preceding twenty-start development screen compared curvature time constants $0,1,2,4$ ms.
Selection required all development paths to settle with clean numerical decay, a median paired variation ratio at most $0.75$, and a cost ratio at most $1.1$; eligible settings were ranked by variation, then cost, then time constant before drawing the fresh validation states.
The initialization rule, model, force limit, and remaining coefficients were then held fixed.

Simulations use Regelum 0.3.15~\cite{Malaniya2025RegelumGraphDep}, with CasADi~\cite{Andersson2019CasADiSoftwareF} CVODES plant integration at relative/absolute tolerances $10^{-10}/10^{-12}$ for the sampled controller.
The continuous augmented state--parameter flow is integrated independently by DOP853 with $10^{-7}/10^{-9}$ tolerances and maximum step $0.02$ s.
The horizon is 45 s.
A path is settled when its last two seconds remain within $|s|<(0.1,0.2,0.05,0.1)$; settling time is the earliest entry remaining inside through the horizon, with at least two seconds left.

The held implementation retains the instantaneous halfspace and adds \eqref{eqn_endpoint}.
At fixed $y$, the endpoint expression is quadratic in $L^+$ and exponential in $\eta^+$, hence convex.
A conic formulation uses second-order and exponential cones and Clarabel~\cite{Goulart2024ClarabelInterio}; rejected numerical candidates may be refined by SLSQP on the same convex constraints.
A final segment repair toward a verified feasible witness is allowed only after checking the original constraints.
This is not a different physical controller, and a repaired answer is not asserted to be the exact minimum-distance optimizer.
Neither refinement nor repair can make a genuinely empty feasible set nonempty.

Numerical success requires full-horizon completion, settling, and both the instantaneous residual and normalized endpoint residual at most $10^{-6}$.
Parameter-box and actual endpoint checks are retained.
These finite tolerances are audit criteria, not replacements for the zero-tolerance inequalities in the theorems.
Independent audits reintegrate recorded held inputs with DOP853 at $10^{-12}/10^{-14}$.
No audited interval with a genuinely positive reproduced violation is relabeled as numerical noise.

\subsection{Validation and comparison}
\label{sub_validation_and_comparison}
Table~\ref{tab_validation} reports the primary outcomes and a second draw.
The second 300-state set, seeds 20261026/27 with the same distributions, was introduced to test a longer-hold candidate.
The unchanged 2 ms controller was retained as its reference and passes all 300 sampled starts.
A separate later evaluation of its unchanged continuous law passes the corresponding 200 box starts.
These are distinct evaluations on reused states, not an additional 200 independent draws.

\begin{table}[t]
\centering
\caption{Nominal cart-pole results.
Strict means the stated numerical decay checks, not a certified region.
Continuous and sampled variants share starts.
The two draw sets contain 600 distinct initial states.}
\label{tab_validation}
\begin{tabular}{llrrr}\toprule
Draw & Controller/evaluation & Starts & Settled & Strict\\\midrule
Primary & CALF, 2 ms held & 300 & 300 & 300\\
Primary & CALF, continuous (box) & 200 & 200 & 200\\
Primary & No-curvature reference, 2 ms & 300 & 300 & 299\\
Primary & NMPC, 100 ms (near hanging) & 100 & 100 & --\\
Second & Unchanged CALF, 2 ms held & 300 & 300 & 300\\
Second & Unchanged CALF, continuous (box) & 200 & 200 & 200\\
Second & Longer-hold candidate, 4 ms & 300 & 300 & 299\\
Second & Longer-hold law, continuous (box) & 200 & 199 & 199\\\bottomrule
\end{tabular}
\end{table}

The primary sample-and-hold candidate/reference median paired command variation ratios are \BoxVariation{} for the box and \NearVariation{} near hanging.
Median paired common-cost ratios are 0.99624 and 0.98584.
This is not uniform improvement: the worst box cost ratio is 2.537 and median box cart travel increases by approximately 1.1\%.
Figure~\ref{fig_trajectory} shows the first near-hanging validation case, selected by its preexisting index rather than by best performance.

\begin{figure}[t]
\centering\includegraphics[width=.95\linewidth]{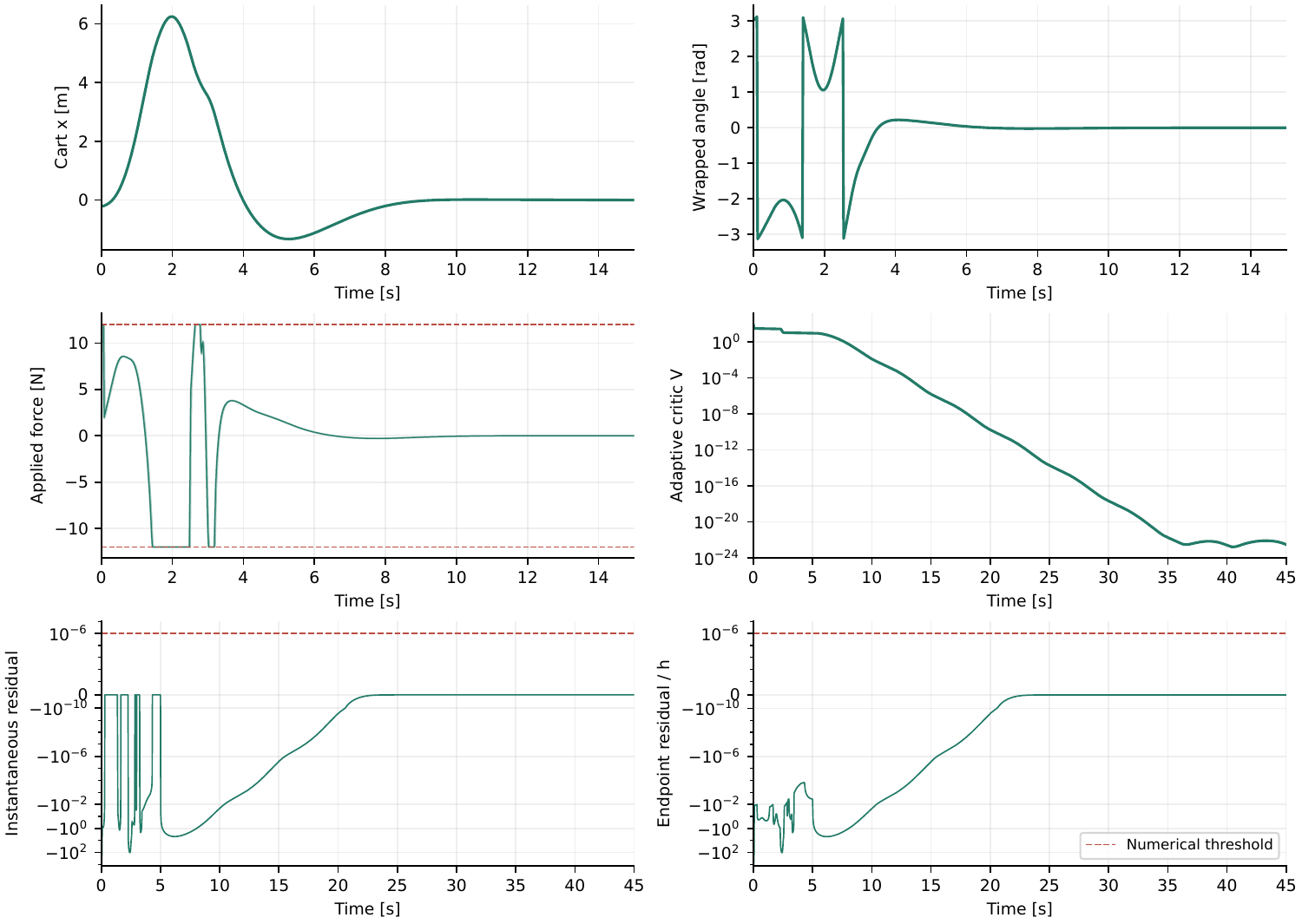}
\caption{First near-hanging validation start with the frozen 12 N CALF.
The first three panels show the maneuver; critic and residual panels cover all 45 s.
Stored 20 ms bins preserve the minimum/maximum force and maximum residual over their ten 2 ms decisions.
The dashed residual threshold is a numerical acceptance criterion, not the theorem's bound.
No energy-based swing-up controller or LQR handoff is used.}
\label{fig_trajectory}
\end{figure}

The matched NMPC comparator uses the same plant and 12 N limit, a 100 ms control period, 30 stages (3 s), multiple shooting, and the common running cost
\begin{equation}
	\label{eqn_commoncost}
 J=\int_0^{45}\{2x^2+v^2+140(1-\cos\vartheta)+8\omega^2+0.35u^2\}\,dt.
\end{equation}
It uses a local terminal quadratic augmented by $70(1-\cos\vartheta)^2$, but no certified terminal set.
Warm starts are shifted after each solve.
At startup, both angle-rotation directions and a stationary seed are considered; a stationary restart is added every ten updates.
Failure of every solve is recorded, not hidden by a backup controller.
The comparator is thus a practical numerical NMPC design, not a globally certified optimal controller.

Both methods settle all 100 matched near-hanging starts.
Median common cost is \CalfCost{} for CALF and \MpcCost{} for NMPC; median settling times are \CalfSettling{} and \MpcSettling{} s, and median peak cart displacements are \CalfTravel{} and \MpcTravel{} m, respectively.
Figure~\ref{fig_comparison} retains the complete distributions.
\begin{figure}[t]
\centering\includegraphics[width=\linewidth]{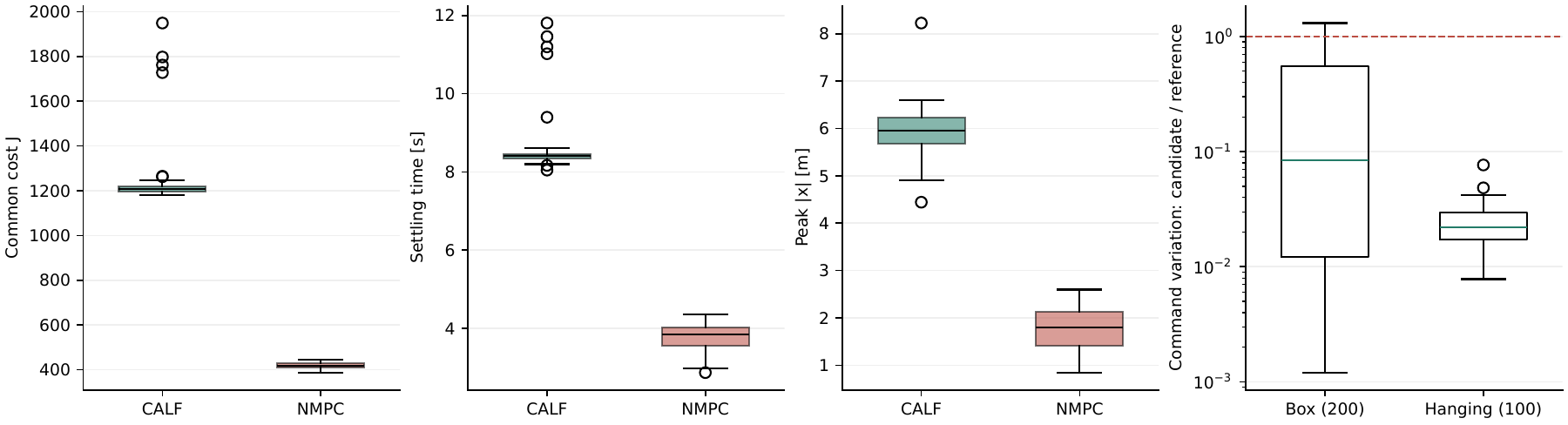}
\caption{Left three panels: all 100 matched near-hanging starts, with lower values preferable.
Right: paired command-variation ratios against the no-curvature CALF reference on all 300 primary starts.
Boxes show quartiles and medians, whiskers extend to 1.5 interquartile ranges, and points beyond are retained.
Cost/travel improvements over NMPC are not claimed.}
\label{fig_comparison}
\end{figure}

\subsection{Audits, computation, and a negative result}
\label{sub_audits_computation_and_a_negative_result}
Thirty-two primary full replays reproduce their archived arrays exactly.
Independent reintegration checks 720000 held intervals: all 16 selected candidate paths pass and 15 of 16 reference paths pass, reproducing the known reference failure.
Maximum state disagreement is below $7.03\times10^{-9}$.
The audits include selected rather than all primary paths; they do not turn the full sampled batch into a validated enclosure.
All 500 primary candidate numerical paths enter the very small local ellipsoid~\eqref{eqn_radius} and stay inside at subsequent logged nodes.
Numerical entry is not an exact-flow proof that swing-up reaches that set.

A quiet single-core benchmark replays six matched initial states for each controller with three sequential repeats.
The median paired total CPU ratio NMPC/CALF is \CpuRatio{}.
The respective median CPU rates are 0.3819 and 0.2277 s per simulated second.
This measures different controllers at their respective 100 ms and 2 ms update periods, not equal-cost solves.
Median steady deadline-miss fractions are 0.1114\% and 0.7023\%; worst steady calls are 431.43 and 4.93 ms.
The simulated plant does not include these computation delays.
Neither hard-real-time operation nor a worst-case execution-time bound is established.

The longer-hold extension is important counterevidence.
A 4 ms hold with a 4 ms curvature penalty passed all 20 reused development starts but failed the next draw: one sampled path had two reproduced endpoint violations despite physically settling, and one continuous path lost feasibility.
The maximum sampled residual was 140.58064, not a roundoff-sized excursion.
This candidate was not promoted.
Changing the penalty changes the feedback, so this is not a discretization-convergence test of one fixed continuous law.

\section{Discussion and Conclusion}
\label{sec_discussion_and_conclusion}
The results demonstrate that an adaptive positive critic can support force-limited swing-up and stabilization without encoding mechanical-energy error, learning from replay data, or switching to a separate local controller.
The local Riccati anchor and deterministic initialization remain meaningful prior structure; the method is neither model-free nor independent of design choices.
Its role is to supply an online mechanism for enforcing a specified critic decrease, not to establish value-function optimality.

The exact pointwise test identifies the compensation available from bounded parameter motion.
It does not prove that the available compensation will persist along the future path.
Both the counterexample and the failed longer-hold experiment make this limitation explicit.
Uniform positivity, existence of closed-loop solutions, and persistent feasibility cannot be replaced by critic richness or by successful trajectories.

For the reported plant, the frozen 2 ms recipe has consistent nominal numerical evidence and a conservative ideal local certificate.
The main remaining theoretical gap is a useful large-excursion feasibility or entry certificate.
Robustness to disturbances, mismatch, actuator dynamics, and finite computation delay is also outside the present result.
The comparison with NMPC exhibits a computational/performance trade-off, not general superiority.
These qualifications define a reproducible continuous-time CALF construction and a concrete starting point for broader adaptive-stabilization studies.

\appendix
\section{Endpoint Optimization and Failure Diagnosis}
\label{sec_endpoint_optimization_and_failure_diagnosis}
For fixed predicted endpoint $y$ and increment $d=hv$, write $H_y=H(y)$, $\psi_y=\psi(y)$, $c_y=\varepsilon\norm{z(y)}^2$, and $A_y=W\gamma(y)\xi_\kappa(y)^\top P\xi_\kappa(y)\ge0$.
The endpoint constraint is
\begin{equation}
 H_y\norm{(L+\Delta L)\psi_y}^2+A_y e^{\eta+d_\eta}
       +H_yc_y\le B,
\end{equation}
where $B$ is the decay budget from~\eqref{eqn_endpoint}.
Introduce epigraph variables $t,e$.
Its conic representation is
\begin{align}
 (t+1,\ 2\sqrt{H_y}(L+\Delta L)\psi_y,\ t-1)&\in\mathcal Q,\nonumber\\
 (\eta+d_\eta,1,e)&\in\mathcal K_{\exp},\qquad t+A_ye+H_yc_y\le B,
\end{align}
together with $hl\le d\le hr$ and $p^\top d+ha\le0$.
The objective is $\frac12\norm{d-hv_0}^2$, with the same minimizer as the velocity objective.
The second-order cone enforces $t\ge H_y\norm{(L+\Delta L)\psi_y}^2$.

A useful necessary endpoint test minimizes the critic over the reachable parameter box, omitting the derivative halfspace.
Each row projection $L_i^+\psi_y$ ranges over a scalar interval $[a_i,b_i]$, obtained by componentwise interval extrema.
The minimum squared projection is $\operatorname{dist}(0,[a_i,b_i])^2$.
The anchor minimum uses the smallest reachable $\eta^+$.
Thus the relaxed endpoint minimum is explicit.
If it exceeds $B$, the fixed-input joint update is infeasible.
Evaluation in floating point is a diagnostic, not an interval proof; it also says nothing about every alternative input.

When the box-minimizing witness satisfies both constraints, convexity permits bisection on the segment joining it to a rejected numerical candidate.
Acceptance rechecks the original box, instantaneous inequality, and endpoint inequality.
Otherwise no repair is accepted.
This explains the difference between numerical projection error and an exhausted feasibility budget.

\section{Sufficient Local Bounds}
\label{sec_local}
This appendix states the estimates used by the ancillary exact-arithmetic checker.
They concern the nominal mathematical feedback, not numerical integration.
All norms below are Euclidean or induced norms; upper bounds may be deliberately loose.
Write $r=\norm s\le1$, $d=5$, $B_L=20d\,3^2$, $\bar P=\operatorname{tr}P$, $\underline P=1/\operatorname{tr}(P^{-1})$, $K_*\ge\norm K$, $B_*=(1+\ell^{-1})/M$, and $T=1+1/6+|\kappa|/2$.
Let $E_-\le e^{-7}$ and $E_2\ge e^2$ be positive rational bounds; put $S_*=WE_-$.
These bounds cover both signs of the fixed initial angle bias.

\subsection{Critic and force remainders}
\label{sub_critic_and_force_remainders}
The lift satisfies $\norm z\le2r$, $\norm{Dz}\le2$, $\norm\psi\le\sqrt5r$, $\norm{D\psi}\le\sqrt8$, $H\le2r^2$, and $\norm{\nabla H}\le4r$.
The monomial estimate follows from $\sum_{i\le j}z_i^2z_j^2\le\norm z^4$.
Consequently the generic gradient is at most $C_G r^3$, where $C_G=32\varepsilon+48B_L$.
The anchor coordinates satisfy
\[
 \norm{\xi_\kappa-s}\le|\kappa|r^2/2+r^3/6,\quad
 \norm{D\xi_\kappa-I}\le|\kappa|r+r^2/2,\quad
 \norm{\xi_\kappa}\le Tr.
\]
Differentiating the gated quadratic, separating the gate, coordinate and Jacobian errors, gives anchor gradient coefficients
\[
 A_2=W\bar P|\kappa|(2T+1),\qquad
 A_3=W\bar P\{(1+T^2)/b^2+T+1/3\}.
\]
Thus $\norm{\nabla_s(W\gamma\xi_\kappa^\top P\xi_\kappa)-2WPs} \le A_2r^2+A_3r^3$.
The input field obeys $\norm{g(s)-B}\le C_g r^2$, with
\[
 C_g=m/M^2+\{m/M^2+1/(2M)\}/\ell.
\]
Before the Gauss--Newton penalty is added, the unsaturated force has $|u+Ks|\le U_2r^2+U_3r^3$, where
\begin{align*}
 U_2&=E_2 A_2B_* /(2WR_0),\\
 U_3&=E_2\left\{\bar P C_g/R_0+A_3B_* /(2WR_0)
       +C_GB_* /(2S_*R_0)+2(1+R_f/(S_*R_0))K_*\right\}.
\end{align*}
Here the lower bound $R\ge e^{-2}We^\eta R_0$ cancels the anchor weight before estimating its gradient error.

For the added penalty, direct differentiation of $r_G$ and $\norm{a+b}^2\le2\norm a^2+2\norm b^2$ give $\norm{\Diff_g r_G}^2\le(20\varepsilon+37B_L)B_*^2r^2$.
Indeed $(\Diff_gH)^2/(2H)\le B_*^2$ away from the target, using $1-e^{-t}\ge te^{-t}$; continuity covers the target.
Also $\norm{\Diff_g r_A}^2\le We^\eta\bar P B_*^2$.
The relative added penalty is therefore at most $C_D r^2$, where
\[
 C_D=\frac{2t_d E_2B_*^2}{R_0}
          \left(\bar P+\frac{20\varepsilon+37B_L}{S_*}\right).
\]
The identity $u_{\rm new}+Ks=(R/R_{\rm new})(u+Ks) +(1-R/R_{\rm new})Ks$ replaces $U_3$ by $U_3+C_DK_*$.
Below $U_3$ denotes this enlarged value.

\subsection{Common state bounds}
\label{sub_common_state_bounds}
Set
\begin{align*}
 C_{a0}&=(K_*+mg)/M,\\
 C_a&=K_*m/M^2+m\ell/M+2mg/(3M)+m^2g/M^2,\\
 C_\omega&=(g/6+C_a+C_{a0}/2)/\ell,\\
 F_2&=B_*U_2,\qquad F_3=C_a+C_\omega+B_*U_3.
\end{align*}
Subtracting the upright linearization from the physical equations gives $\norm{f+gk-A_cs}\le F_2r^2+F_3r^3$, provided the force bound below ensures that saturation is inactive.
For $q_c=1/\operatorname{tr}(Q_c^{-1})>0$, sufficient continuous conditions are
\[
 2\bar P(F_2r+F_3r^2)\le q_c/2,\qquad
 K_*r+U_2r^2+U_3r^3\le12.
\]
They make $s^\top Ps\le\underline P r^2$ invariant with strict quadratic state decrease for arbitrary boxed parameter variation.

For the held flow, use $G_h\ge e^{\norm A h}$ and an integer $\beta\ge2G_h\{1+hB_*(K_*+1)\}$.
Let $C_d$ be $C_a+C_\omega$ with $K_*$ replaced by $K_*+1$.
A bootstrap on $[0,h]$ yields
\[
 \norm{s^+-F_hs}\le D_2r^2+D_3r^3,\quad
 D_2=hG_hB_*U_2,\quad D_3=hG_h(B_*U_3+C_d\beta^3),
\]
if $\beta r\le1$, $U_2r+U_3r^2\le1$, and $G_hhC_d\beta^3r^2\le\beta/2$.
In particular $\norm{s(t)}\le\beta\norm{s(0)}$ during the hold.
Let $q_h>0$ lower-bound $\lambda_{\min}(P-F_h^\top PF_h)$ and $H_*\ge\norm{F_h}$.
With $E_h=D_2r+D_3r^2$, the additional condition
\[
 \bar P(2H_*E_h+E_h^2)\le q_h/2
\]
and the same force bound give common sampled state invariance and decay.

\subsection{Feasibility inequalities}
\label{sub_feasibility_inequalities}
After normalization by $We^\eta$, valid gradient and value error bounds are $g_2r^2+g_3r^3$ and $v_3r^3+v_4r^4$, respectively, where
\begin{align*}
 g_2&=\bar P|\kappa|(2T+1),&
 g_3&=\bar P\{(1+T^2)/b^2+T+1/3\}+C_G/S_*,\\
 v_3&=\bar P|\kappa|(T+1)/2,&
 v_4&=\bar P\{1/(2b^2)+(T+1)/6\}+(8\varepsilon+10B_L)/S_*.
\end{align*}
Also $\norm{\nabla_LV}\le20\sqrt{B_L}r^4$ and $\norm{\dot L}\le2\sqrt{B_L}/\tau$.
Let $m_c>0$ and $m_h>0$ be verified lower eigenvalue bounds for the continuous feasibility matrix and~\eqref{eqn_samplematrix}, using $S_*$ in place of $S_-$ and $\zeta=\delta/10$.
This reserve covers the solver's additional margin $\min(10^{-4},\delta q/10)$.
Let $J_*\ge\norm{A_c}$ and $T_*\ge\norm{P+h\sigma Q_s/(2S_*)}$.
The sufficient tests are
\begin{align*}
 2\bar P E_c+(g_2r+g_3r^2)(J_*+E_c)
          +40B_Lr^2/(\tau S_*)&\le m_c/2,\quad E_c=F_2r+F_3r^2,\\
 T_*(2H_*E_h+E_h^2)+v_3(1+\beta^3)r
          +v_4(1+\beta^4)r^2&\le m_h/2.
\end{align*}
The first bounds the normalized total derivative with the anchor fixed, even for arbitrary allowed factor velocities.
The second bounds both endpoint critic errors independently for arbitrary boxed factors with the anchor fixed.
Thus zero parameter velocity is a feasible witness.
Combining these tests with the common state bounds preserves a joint local state--parameter region for every verified feasible boxed update.

\subsection{Exact numerical check}
\label{sub_exact_numerical_check}
The ancillary implementation evaluates these constants as fractions.
For any positive-definite rational matrix $Z$, exact symmetric elimination checks its positive pivots, and $1/\operatorname{tr}(Z^{-1})$ is a valid lower eigenvalue bound.
Degree-12 Taylor polynomials enclose $F_h$; if $\norm{F_h-\widehat F_h}\le e_h$, the pullback error is at most $\norm{T}(2\norm{\widehat F_h}e_h+e_h^2)$ for the corresponding matrix $T$.
Positive scalar exponentials use the rational bound
\[
 e^a\le\sum_{j=0}^N\frac{a^j}{j!}
       +\frac{a^{N+1}}{(N+1)!}\frac1{1-a/(N+2)},\qquad0\le a<N+2.
\]
Reciprocation supplies $E_-$.
Halving a rational radius until every sufficient inequality holds produces~\eqref{eqn_radius}.
The checker recomputes every stored final constant and rejects a changed matrix or arithmetic implementation.
It verifies this calculation conditional on the analytic estimates, not a formal proof of the software or a validated enclosure of the preceding swing-up.

\section{Reproducibility and Claim Boundaries}
\label{sec_reproducibility_and_claim_boundaries}
The ancillary material contains the frozen cart-pole configuration, all compact primary and second-draw outcome tables, the supplementary continuous evaluation, matched NMPC and paired-reference tables, independent audit summaries, timing summaries, and exact local-certificate constants.
The figure script regenerates the two figures from those data alone.
A separate checker recomputes the exact local constants using the archived binary matrix, without dependence on a particular floating-point Riccati solver.
The representative trajectory is a compact reduction of a full audited replay; its bin maxima preserve all positive residual excursions.
Source hashes identify the underlying archived data.
No double-pendulum or other plant result is used in this preprint.

There are 600 distinct sampled initial conditions, not 1000 independent states obtained by adding 400 continuous evaluations.
The matched NMPC table repeats 100 of the primary states, and replay/audit records repeat selected paths.
A zero observed failure count is not a proof of a region or a probability-one guarantee.
The local exact-arithmetic calculation concerns an ideal nominal controller and cannot be promoted to a machine-arithmetic or validated full-swing-up certificate.

\section{Note on AI Use}
\label{sec_note_on_ai_use}
AI was used to assist with the case study and to polish the mathematics and language.
The idea and approach were created by the author.
Prototypical forms appeared in the author's earlier work~\cite{Osinenko2022StabilizingRein}; see also the published actor-critic framework~\cite{Osinenko2023ActorCriticFram}.
The motivation was to develop CALF (Critic as Lyapunov Function), an approach created by the author and defended in his habilitation, to remove the need for a fallback policy.
The current work marks a milestone in this development.
Further case studies and research are ongoing.

\begingroup
\let\standardbibliography\thebibliography
\renewcommand{\thebibliography}[1]{%
  \standardbibliography{#1}%
  \setlength{\itemsep}{0pt}%
}
\bibliographystyle{ieeetr}
\bibliography{AIDA}

@Article{Sontag1989UniversalConstr,
  author  = {Sontag, Eduardo D.},
  title   = {A `universal' construction of {Artstein}'s theorem on nonlinear stabilization},
  journal = {Systems \& Control Letters},
  year    = {1989},
  volume  = {13},
  number  = {2},
  pages   = {117--123},
  doi     = {10.1016/0167-6911(89)90028-5},
  note    = {\href{https://doi.org/10.1016/0167-6911(89)90028-5}{doi:10.1016/0167-6911(89)90028-5}},
}

@Article{Krstic1995ControlLyapunov,
  author  = {Krsti\'{c}, Miroslav and Kokotovi\'{c}, Peter V.},
  title   = {Control {Lyapunov} functions for adaptive nonlinear stabilization},
  journal = {Systems \& Control Letters},
  year    = {1995},
  volume  = {26},
  number  = {1},
  pages   = {17--23},
  doi     = {10.1016/0167-6911(94)00107-7},
  note    = {\href{https://doi.org/10.1016/0167-6911(94)00107-7}{doi:10.1016/0167-6911(94)00107-7}},
}

@Article{Osinenko2020ReinforcementL,
  author  = {Osinenko, Pavel and Beckenbach, Lukas and G\"{o}hrt, Thomas and Streif, Stefan},
  title   = {A reinforcement learning method with closed-loop stability guarantee},
  journal = {IFAC-PapersOnLine},
  year    = {2020},
  volume  = {53},
  number  = {2},
  pages   = {8043--8048},
  doi     = {10.1016/j.ifacol.2020.12.2237},
  note    = {\href{https://doi.org/10.1016/j.ifacol.2020.12.2237}{doi:10.1016/j.ifacol.2020.12.2237}},
}

@Article{Osinenko2023ActorCriticFram,
  author  = {Osinenko, Pavel and Yaremenko, Grigory and Malaniya, Georgiy and Bolychev, Anton},
  title   = {An actor-critic framework for online control with environment stability guarantee},
  journal = {IEEE Access},
  year    = {2023},
  volume  = {11},
  pages   = {89188--89204},
  doi     = {10.1109/ACCESS.2023.3306070},
  note    = {\href{https://doi.org/10.1109/ACCESS.2023.3306070}{doi:10.1109/ACCESS.2023.3306070}},
}

@InProceedings{Osinenko2024CriticAsLyapuno,
  author    = {Osinenko, Pavel and Yaremenko, Grigory and Zashchitin, Roman and Bolychev, Anton and Ibrahim, Sinan and Dobriborsci, Dmitrii},
  title     = {Critic as {Lyapunov} function ({CALF}): a model-free, stability-ensuring agent},
  booktitle = {2024 IEEE 63rd Conference on Decision and Control (CDC)},
  year      = {2024},
  pages     = {2517--2524},
  publisher = {IEEE},
  doi       = {10.1109/CDC56724.2024.10886034},
  note      = {\href{https://doi.org/10.1109/CDC56724.2024.10886034}{doi:10.1109/CDC56724.2024.10886034}},
}

@Article{Hauswirth2021ProjectedDynami,
  author  = {Hauswirth, Adrian and Bolognani, Saverio and D\"{o}rfler, Florian},
  title   = {Projected dynamical systems on irregular, non-{Euclidean} domains for nonlinear optimization},
  journal = {SIAM Journal on Control and Optimization},
  year    = {2021},
  volume  = {59},
  number  = {1},
  pages   = {635--668},
  doi     = {10.1137/18M1229225},
  note    = {\href{https://doi.org/10.1137/18M1229225}{doi:10.1137/18M1229225}},
}

@InProceedings{Malaniya2025RegelumGraphDep,
  author    = {Malaniya, Georgiy and Bolychev, Anton and Osinenko, Pavel and Yaremenko, Grigory},
  title     = {{Regelum}: graph dependency resolution and execution orchestration for control systems},
  booktitle = {2025 IEEE 26th International Conference of Young Professionals in Electron Devices and Materials (EDM)},
  year      = {2025},
  pages     = {1470--1475},
  publisher = {IEEE},
  doi       = {10.1109/EDM65517.2025.11096894},
  note      = {\href{https://doi.org/10.1109/EDM65517.2025.11096894}{doi:10.1109/EDM65517.2025.11096894}},
}

@Article{Andersson2019CasADiSoftwareF,
  author  = {Andersson, Joel A. E. and Gillis, Joris and Horn, Greg and Rawlings, James B. and Diehl, Moritz},
  title   = {{CasADi}: a software framework for nonlinear optimization and optimal control},
  journal = {Mathematical Programming Computation},
  year    = {2019},
  volume  = {11},
  number  = {1},
  pages   = {1--36},
  doi     = {10.1007/s12532-018-0139-4},
  note    = {\href{https://doi.org/10.1007/s12532-018-0139-4}{doi:10.1007/s12532-018-0139-4}},
}

@Misc{Goulart2024ClarabelInterio,
  author        = {Goulart, Paul J. and Chen, Yuwen},
  title         = {{Clarabel}: an interior-point solver for conic programs with quadratic objectives},
  year          = {2024},
  eprint        = {2405.12762},
  archivePrefix = {arXiv},
  primaryClass  = {math.OC},
  howpublished  = {\href{https://arxiv.org/abs/2405.12762}{arXiv:2405.12762}},
}

@Misc{Osinenko2022StabilizingRein,
  author        = {Osinenko, Pavel and Yaremenko, Grigory and Malaniya, Georgiy},
  title         = {On stabilizing reinforcement learning without {Lyapunov} functions},
  year          = {2022},
  eprint        = {2207.08730v2},
  archivePrefix = {arXiv},
  primaryClass  = {eess.SY},
  howpublished  = {\href{https://arxiv.org/abs/2207.08730v2}{arXiv:2207.08730v2}},
}

@Article{Vamvoudakis2010OnlineActorCrit,
  author  = {Vamvoudakis, Kyriakos G. and Lewis, Frank L.},
  title   = {Online actor--critic algorithm to solve the continuous-time infinite horizon optimal control problem},
  journal = {Automatica},
  year    = {2010},
  volume  = {46},
  number  = {5},
  pages   = {878--888},
  doi     = {10.1016/j.automatica.2010.02.018},
  note    = {\href{https://doi.org/10.1016/j.automatica.2010.02.018}{doi:10.1016/j.automatica.2010.02.018}},
}

@InProceedings{Dierks2010OptimalControlO,
  author    = {Dierks, T. and Jagannathan, S.},
  title     = {Optimal control of affine nonlinear continuous-time systems},
  booktitle = {Proceedings of the 2010 American Control Conference},
  year      = {2010},
  pages     = {1568--1573},
  publisher = {IEEE},
  doi       = {10.1109/ACC.2010.5531586},
  note      = {\href{https://doi.org/10.1109/ACC.2010.5531586}{doi:10.1109/ACC.2010.5531586}},
}

@Article{Lopez2022UniversalAdapti,
  author  = {Lopez, Brett T. and Slotine, Jean-Jacques E.},
  title   = {Universal adaptive control of nonlinear systems},
  journal = {IEEE Control Systems Letters},
  year    = {2022},
  volume  = {6},
  pages   = {1826--1830},
  doi     = {10.1109/LCSYS.2021.3133359},
  note    = {\href{https://doi.org/10.1109/LCSYS.2021.3133359}{doi:10.1109/LCSYS.2021.3133359}},
}

@Article{Sinaei2026SafeAdaptiveCon,
  author  = {Sinaei, Kasra and Ebeigbe, Donald},
  title   = {Safe adaptive control with vanishing conservativeness for robotic systems with unknown dynamics via barrier functions},
  journal = {IEEE Robotics and Automation Letters},
  year    = {2026},
  volume  = {11},
  number  = {5},
  pages   = {6161--6168},
  doi     = {10.1109/LRA.2026.3678137},
  note    = {\href{https://doi.org/10.1109/LRA.2026.3678137}{doi:10.1109/LRA.2026.3678137}},
}

@InProceedings{Chow2018ALyapunovBasedA,
  author    = {Chow, Yinlam and Nachum, Ofir and Duenez-Guzman, Edgar and Ghavamzadeh, Mohammad},
  title     = {A {Lyapunov}-based approach to safe reinforcement learning},
  booktitle = {Advances in Neural Information Processing Systems},
  year      = {2018},
  volume    = {31},
  publisher = {Curran Associates, Inc.},
  url       = {https://papers.neurips.cc/paper_files/paper/2018/hash/4fe5149039b52765bde64beb9f674940-Abstract.html},
  note      = {\href{https://papers.neurips.cc/paper_files/paper/2018/hash/4fe5149039b52765bde64beb9f674940-Abstract.html}{NeurIPS 2018 proceedings}},
}

@Article{Han2020ActorCriticRein,
  author  = {Han, Minghao and Zhang, Lixian and Wang, Jun and Pan, Wei},
  title   = {Actor-critic reinforcement learning for control with stability guarantee},
  journal = {IEEE Robotics and Automation Letters},
  year    = {2020},
  volume  = {5},
  number  = {4},
  pages   = {6217--6224},
  doi     = {10.1109/LRA.2020.3011351},
  note    = {\href{https://doi.org/10.1109/LRA.2020.3011351}{doi:10.1109/LRA.2020.3011351}},
}

@Article{Khaniki2026SafeAndAdaptive,
  author  = {Khaniki, Mohammad Ali Labbaf and Mirzaee, Morteza and Moradi, Elahe},
  title   = {Safe and adaptive control of non-stationary stochastic systems via {Lyapunov}-constrained distributional reinforcement learning},
  journal = {Scientific Reports},
  year    = {2026},
  volume  = {16},
  number  = {1},
  eid     = {26184},
  doi     = {10.1038/s41598-026-66485-2},
  note    = {Article 26184. \href{https://doi.org/10.1038/s41598-026-66485-2}{doi:10.1038/s41598-026-66485-2}},
}

@InProceedings{Chang2019NeuralLyapunovC,
  author    = {Chang, Ya-Chien and Roohi, Nima and Gao, Sicun},
  title     = {Neural {Lyapunov} control},
  booktitle = {Advances in Neural Information Processing Systems},
  year      = {2019},
  volume    = {32},
  publisher = {Curran Associates, Inc.},
  url       = {https://papers.neurips.cc/paper_files/paper/2019/hash/2647c1dba23bc0e0f9cdf75339e120d2-Abstract.html},
  note      = {\href{https://papers.neurips.cc/paper_files/paper/2019/hash/2647c1dba23bc0e0f9cdf75339e120d2-Abstract.html}{NeurIPS 2019 proceedings}},
}

@InProceedings{Mehrjou2021NeuralLyapunovR,
  author    = {Mehrjou, Arash and Ghavamzadeh, Mohammad and Sch\"{o}lkopf, Bernhard},
  title     = {Neural {Lyapunov} redesign},
  booktitle = {Proceedings of the 3rd Conference on Learning for Dynamics and Control},
  year      = {2021},
  volume    = {144},
  series    = {Proceedings of Machine Learning Research},
  pages     = {459--470},
  publisher = {PMLR},
  url       = {https://proceedings.mlr.press/v144/mehrjou21a.html},
  note      = {\url{https://proceedings.mlr.press/v144/mehrjou21a.html}},
}

@InProceedings{Yang2024LyapunovStableN,
  author    = {Yang, Lujie and Dai, Hongkai and Shi, Zhouxing and Hsieh, Cho-Jui and Tedrake, Russ and Zhang, Huan},
  title     = {{Lyapunov}-stable neural control for state and output feedback: a novel formulation},
  booktitle = {Proceedings of the 41st International Conference on Machine Learning},
  year      = {2024},
  volume    = {235},
  series    = {Proceedings of Machine Learning Research},
  pages     = {56033--56046},
  publisher = {PMLR},
  url       = {https://proceedings.mlr.press/v235/yang24f.html},
  note      = {\url{https://proceedings.mlr.press/v235/yang24f.html}},
}

@InProceedings{Li2025TwoStageLearnin,
  author    = {Li, Haoyu and Zhong, Xiangru and Hu, Bin and Zhang, Huan},
  title     = {Two-stage learning of stabilizing neural controllers via {Zubov} sampling and iterative domain expansion},
  booktitle = {Advances in Neural Information Processing Systems},
  year      = {2025},
  volume    = {38},
  pages     = {73980--74009},
  publisher = {Curran Associates, Inc.},
  doi       = {10.52202/085713-2480},
  note      = {\href{https://doi.org/10.52202/085713-2480}{doi:10.52202/085713-2480}},
}

@Article{Ames2017ControlBarrierF,
  author  = {Ames, Aaron D. and Xu, Xiangru and Grizzle, Jessy W. and Tabuada, Paulo},
  title   = {Control barrier function based quadratic programs for safety critical systems},
  journal = {IEEE Transactions on Automatic Control},
  year    = {2017},
  volume  = {62},
  number  = {8},
  pages   = {3861--3876},
  doi     = {10.1109/TAC.2016.2638961},
  note    = {\href{https://doi.org/10.1109/TAC.2016.2638961}{doi:10.1109/TAC.2016.2638961}},
}
\endgroup
\end{document}